\def\year{2022}\relax
\documentclass[letterpaper]{article} 
\usepackage{aaai22}  
\usepackage{times}  
\usepackage{helvet}  
\usepackage{courier}  
\usepackage[hyphens]{url}  
\usepackage{graphicx} 
\usepackage{natbib}  
\usepackage{caption} 
\DeclareCaptionStyle{ruled}{labelfont=normalfont,labelsep=colon,strut=off} 
\usepackage{algorithm}
\usepackage{algpseudocode}
\usepackage{syntax}
\usepackage{amsmath}
\usepackage{stmaryrd}
\usepackage{amsthm}
\usepackage{mathpartir}
\newtheorem{theorem}{Theorem}
\newtheorem{lemma}[theorem]{Lemma}
\newtheorem{proposition}[theorem]{Proposition}
\theoremstyle{definition}
\newtheorem{definition}[theorem]{Definition}
\usepackage{soul}
\usepackage{newfloat}
\usepackage{listings}
\floatstyle{ruled}
\newfloat{listing}{tb}{lst}{}
\floatname{listing}{Listing}
\usepackage{color} 

\title{A Variant of Concurrent Constraint Programming on GPU}
\author{
    Pierre Talbot, Fr\'{e}d\'{e}ric Pinel, Pascal Bouvry
}
\affiliations{
    Interdisciplinary Centre for Security, Reliability and Trust (SnT) \\
    University of Luxembourg\\
    Esch-sur-Alzette L-4243, Luxembourg

    pierre.talbot@uni.lu, frederic.pinel@uni.lu, pascal.bouvry@uni.lu
}

\usepackage{amssymb}
\newcommand\eqdef{\,\triangleq\,}

\begin{document}

\maketitle

\begin{abstract}
The number of cores on graphical computing units (GPUs) is reaching thousands nowadays, whereas the clock speed of processors stagnates.
Unfortunately, constraint programming solvers do not take advantage yet of GPU parallelism.
One reason is that constraint solvers were primarily designed within the mental frame of sequential computation.
To solve this issue, we take a step back and contribute to a simple, intrinsically parallel, lock-free and formally correct programming language based on \textit{concurrent constraint programming}.
We then re-examine parallel constraint solving on GPUs within this formalism, and develop \textsc{Turbo}, a simple constraint solver entirely programmed on GPUs.
\textsc{Turbo} validates the correctness of our approach and compares positively to a parallel CPU-based solver.
\end{abstract}

\section{Introduction}

The number of cores of graphical computing units (GPUs) steadily increases over the years, reaching thousands nowadays, whereas central processing unit (CPU) clock speed and core numbers stagnate~\cite{decline-cpus-2021}.
However, it is notoriously difficult and error prone to write parallel programs on shared memory systems~\cite{lee-threads}.
Threads introduce nondeterminism that must be pruned using synchronization mechanisms such as locks and barriers.
However, too much synchronization quickly reduce the efficiency of a parallel program, sometimes making it slower than its sequential counterpart.

This phenomenon is observed in the field of \textit{constraint programming}, where constraint solvers do not benefit yet from the parallelism offered by GPUs.
One of the reasons is that most of the key optimisations of constraint solvers, including propagation loop, search strategies, global constraints and various forms of constraints learning, were primarily designed in the mental frame of sequential computation.

To grasp the challenge, we consider a standard solving algorithm named \textit{propagate and search} (see, \emph{e.g.},~\citet{propagation-guido-tack}).
Each constraint is implemented by a function, called a \textit{propagator}, which removes inconsistent values from the domain of its variables.
The propagation is the process of executing the propagators until a fixed point is reached and no more values can be removed.
The \textit{search step} then splits the problem into two complementary subproblems.
Propagation and search alternates recursively until either a solution is found or no solution could be found.

Search can be parallelized without much hurdles by executing each subproblem on a different thread.
This approach has been explored in-depth on CPUs and in distributed settings~\cite{perron-search-1999,schulte-simple-pcp,malapert-embarrassingly-2016}.
However, it is not very efficient to run a subproblem per thread on a GPU; in part due to the limited cache memory available per thread.
On the other hand, propagation is almost never parallelized in modern constraint solvers because propagators are executed in a highly sequential manner~\cite{schulte-efficient-2008}.
Moreover, as propagators share variables, it seems that synchronization will unavoidably hinder efficiency and complicate the solver architecture~\cite{gent-review-2018}.

In this paper, we contribute to a simple, intrinsically parallel, lock-free and formally correct programming language based on \textit{concurrent constraint programming} (CCP)~\cite{saraswat-phd}, that we call parallel CCP (PCCP).
CCP has been very successful to describe concurrent computations at the theoretical level, but has had limited impact on practical implementations.
PCCP aims to bridge this gap between concurrency theory and parallel programming.

The key characteristics of PCCP is to allow threads to progress and communicate on a shared memory without synchronization, and still to always reach a unique deterministic result.
The main ingredient is the notion of fixed point over \textit{lattices}~\cite{davey-introduction-2002}, a practical and mathematical structure, which guarantees the correct combination of parallel processes.
To show the correctness of our language, we successively refine the semantics of PCCP from denotational, operational and then to a low-level semantics based on load and store operations.
In particular, we show that whether the fixed point of a PCCP program is computed sequentially or in parallel, the result of the execution remains the same.

Through the lens of PCCP, we then re-examine parallel constraint solving on GPUs.
We propose \textsc{Turbo}, the first propagate and search constraint solver fully implemented on GPUs.
To accommodate the architecture of GPUs, we have made several algorithmic simplifications which are normally considered crucial for efficiency in sequential solvers.
For instance, our propagation loop is eventless and is reminiscent of the AC-1 algorithm~\cite{mackworth-consistency-1977}.
Our experiments show that \textsc{Turbo} obtains slightly better performance when compared to \textsc{GeCode}~\cite{gecode}, a parallel CPU-based solver.
\textsc{Turbo} is only a proof of concept and several key components of modern constraint solvers are missing, in particular global constraints and constraint learning.
Nevertheless, constraint decomposition of global constraints is an efficient approach that can be immediately used in \textsc{Turbo}~\cite{narodytska-reformulation-2011,bessiere-range-roots-2009,schutt-why-2009}.
Still, it is clear that many improvements will be required to make \textsc{Turbo} a competitive solver.

\section{Parallel Concurrent Constraint Programming}

\textit{Parallel concurrent constraint programming} (PCCP) is a variant of determinate CCP~\cite{saraswat-semantic-1991} adapted to parallel programming.
The main changes are an explicit usage of lattice-typed variables instead of a constraint system, and a reduced set of instructions (essentially without recursive procedures).

A \textit{lattice} is an ordered structure $\langle L, \leq \rangle$ where $L$ is a set and $\leq: L \times L$ a partial order relation.
The \textit{join operation} $\sqcup: L \times L \to L$ is obtained by $x \leq y \Leftrightarrow x \sqcup y = y$.
When they exist, we write $\bot$ the smallest element and $\top$ the greatest element in $L$.
The dual lattice of $L$, written $L^\partial$, is the lattice with the reversed partial order ($\geq$ instead of $\leq$) on the same set of elements.
A function $f: L \to L$ is \textit{extensive} on $L$ if $x \leq f(x)$, \textit{monotone} if $x \leq y \Rightarrow f(x) \leq f(y)$ for all $x,y \in L$ and \textit{idempotent} if $f(x) = f(f(x))$.
Closure operators are extensive, monotone and idempotent functions.
The pointwise lifting is a lattice construction ordering functions $f,g \in L \to K$ where $f \leq g$ iff $\forall{x \in L},~f(x) \leq g(x)$.
The pointwise join is defined as $f \sqcup g \eqdef \lambda x.f(x) \sqcup g(x)$.
A fixed point of $f$ is an element $x \in L$ such that $f(x) = x$.
We write $\mathbf{fp}_x(f)$ the fixed point $f(f(...f(x)))$, and $\mathbf{fix}~f$ the function $\lambda x.\mathbf{fp}_x(f)$ mapping any element $x$ to its fixed point.
Details on lattice theory can be found in~\cite{davey-introduction-2002}.

We now overview several useful lattices and the Cartesian product of two lattices.
Firstly, we have the lattice of \textit{increasing integers} $\mathit{ZInc} = \langle \mathit{Z} \cup \{-\infty, \infty\}, \leq \rangle$, such that $\mathit{Z} \subset \mathbb{Z}$.
The order is given by the natural arithmetic order, the smallest element is $-\infty$ and the greatest one is $\infty$.
The join operation is defined as the maximum of the two elements: $a \sqcup b \eqdef \mathit{max}(a, b)$.
Similarly, we can define $\mathit{ZDec}$ for the decreasing integers (note that $\mathit{ZDec} = \mathit{ZInc}^\partial$), $\mathit{FInc}$ and $\mathit{FDec}$ for the increasing and decreasing floating-point numbers\footnote{We can follow the total order given by IEEE-754-2019~\cite{ieee-754-2019}.}, and $\mathit{BDec}$ and $\mathit{BInc}$ for the lattice of Boolean where $\{\mathit{true}, \mathit{false}\}$ where $\mathit{true} \leq \mathit{false}$ in $\mathit{BDec}$.
We note that these lattices form chains.
In order to obtain new lattices from these primitive ones, the Cartesian product of two lattices $L \times K$ can be used.
In this case, the order and join are defined pointwise, \emph{i.e.}, $(a,b) \leq (a',b') \Leftrightarrow a \leq a'$ and $b \leq b'$, and $(a,b) \sqcup (a',b') \eqdef (a \sqcup a', b \sqcup b')$.
The Cartesian product can be generalized to a $n$-ary product $L_1 \times \ldots \times L_n$.
It is equipped with a monotone projection function $\pi_i((a_1,\ldots,a_n)) \eqdef a_i$ and a monotone embedding function $\mathit{embed}_i((a_1,\ldots,a_n), b_i) \eqdef (a_1,\ldots,a_i \sqcup b_i,\ldots, a_n)$.
This way, we can obtain the lattice of integer intervals $\mathit{IZ} = \mathit{ZInc} \times \mathit{ZDec}$ where an element $(\ell, u)$ represents the set of values $\{v \in \mathbb{Z} \;|\; \ell \leq v \leq u\}$.
The lattice of intervals forms a partial order that is not a chain.
We write the monotone projection functions $\pi_1$ and $\pi_2$ more explicitly by $\lceil (\ell, u) \rceil \eqdef u$ and $\lfloor (\ell, u) \rfloor \eqdef \ell$.

Lattices are types in PCCP, and elements of these lattices are the data manipulated by PCCP programs.
The syntax of PCCP is defined as follows, where $x,y,y_1,\ldots,y_n \in \mathit{Vars}$ are variables, $L$ a lattice, $f$ a monotone function, and $b$ a Boolean variable of type $\mathit{BInc}$:

\begin{grammar}
<P, Q> ::= \texttt{if $b$ then $P$} \hfill \textit{ask statement}
\alt $x \leftarrow f(y_1,...,y_n)$ \hfill \textit{tell statement}
\alt $\exists{x\mathtt{\mathord{:}L}},~P$ \hfill \textit{local statement}
\alt $P\;\|\;Q$ \hfill \textit{parallel composition}
\end{grammar}
\noindent
The ask statement tests the value of $b$, and if it is equal to $\mathit{true}$, executes the process $P$.
The tell statement joins the value of the variable $x$ and the result of a monotone function $f$.
We declare a new variable using the local statement $\exists{x\mathtt{\mathord{:}L}},~P$ which initializes a variable $x$ of type $\mathtt{L}$ to $\bot$, such that $x$ is only visible to $P$.
The parallel composition executes two processes $P$ and $Q$ in parallel.
There is no explicit sequential composition, but sequentiality is still present due to ask statements---some processes might not be able to execute until some information is provided by other processes.
A useful syntactic sugar is to allow an expression $e$ in an ask statement $\mathtt{if}~e~\mathtt{then}~P$, which is equivalent to $\exists{b\mathtt{:}\mathit{BInc}},~b \leftarrow e \;\|\; \mathtt{if}~b~\mathtt{then}~P$ where $b$ is not free in $P$.

PCCP instructions are purposefully low-level because they should be directly implementable in a programming language such as C++ or Java.
For clarity, we introduce a light modelling layer on top of PCCP, which generates PCCP programs at compile-time.
Firstly, we use a generator expanding parallel processes at compile-time, \emph{e.g.}, $\forall{i \in \{1,2\}},~x_i \leftarrow \bot$ expands to $x_1 \leftarrow \bot \;\|\; x_2 \leftarrow \bot$.
We sometimes use generators inside expressions such as in  $\sum_{i \in S}~x_i$.
All indices in names are resolved at compile-time, hence $x_1$ is really just the name of a variable.
Our usage of generators is similar to those in traditional modelling languages such as \textsf{MiniZinc}~\cite{nethercote-minizinc:-2007}, and their meanings should be clear from the context.
Secondly, we compile constraints into corresponding PCCP programs using the function $\llbracket . \rrbracket: \Phi \to \mathit{Proc}$ where $\Phi$ is the set of constraints and $\mathit{Proc}$ the set of PCCP processes.
The obtained process $\llbracket \varphi \rrbracket$ is called a \textit{propagator} for the constraint $\varphi$.
Finally, we rely on the monotone function $\mathit{entailed}: \Phi \to \mathit{BInc}$ which maps a formula $\varphi$ to $\mathit{true}$ whenever $\varphi$ is entailed from the store, and to $\mathit{false}$ otherwise.

We illustrate PCCP on the resource-constrained project scheduling problem (RCPSP).
RCPSP is an optimization problem where one must find a minimal schedule such that resources usage do not exceed some capacities and some precedence constraints are satisfied.
RCPSP is defined by a tuple $\langle T, P, R \rangle$ where $T$ is the set of tasks, $P$ is the set of precedences among tasks, written $i \ll j$ to indicate that task $i$ must terminate before $j$ starts, and $R$ is the set of resources.
Each task $i \in T$ has a duration $d_i \in \mathbb{N}$ and, for each resource $k \in R$, a resource usage $r_{k,i} \in \mathbb{N}$.
Each resource $k \in R$ has a capacity $c_k \in \mathbb{N}$ quantifying how much of this resource is available in each instant.
The decision variables are the starting date $s_i$ for each task $i$, and $n^2$ Boolean variables $b_{i,j}$ such that $b_{i,j}$ is true iff the tasks $i$ and $j$ overlap.
This is a standard model of RCPSP which decomposes the cumulative global constraint~\cite{schutt-why-2009}.
The PCCP model for RCPSP is given as follows (where $h$ is the horizon, \emph{e.g.}, the sum of the durations):
\begin{displaymath}
\begin{array}{l}
\exists{s_1:\mathit{IZ}},\ldots,\exists{s_n:\mathit{IZ}},\\
\exists{b_{1,1}:\mathit{IZ}},\ldots,\exists{b_{n,n}:\mathit{IZ}},\\
\phantom{ \;\|\; } \forall{i \in [1..n]},~s_i \leftarrow (0, h) \\
\;\|\; \forall{i \in [1..n]},~\forall{j \in [1..n]},~b_{i,j} \leftarrow (0, 1) \\
\;\|\; \forall{(i \ll j) \in P},~\llbracket s_i + d_i \leq s_j \rrbracket \\
 \;\|\; \forall{j \in [1..n]},~\forall{i \in [1..n]},\\
\phantom{ \;\|\; }\quad \llbracket b_{i,j} \Leftrightarrow (s_i \leq s_j \land s_j < s_i + d_i) \rrbracket \\
 \;\|\; \forall{k \in R},~\forall{j \in [1..n]},~\llbracket \sum_{i \in [1..n]} r_{k,i} * b_{i,j} \leq c_k \rrbracket
\end{array}
\end{displaymath}
The compilation of constraints into PCCP processes is reminiscent of the compilation of constraints into indexicals, a simple CCP constraint system~\cite{cc-fd-impl,carlsson-indexicals-1997}.
The precedence and resource constraints are compiled as follows:
\begin{displaymath}
\begin{array}{l}
\llbracket x + y \leq c \rrbracket \eqdef x \leftarrow (\bot, c - \lfloor y \rfloor) \;\|\; y \leftarrow (\bot, c - \lfloor x \rfloor) \\[0.3cm]
\llbracket \sum_{i \in [1..n]} r_{k,i} * b_{i,j} \leq c_k \rrbracket \eqdef \exists{\mathit{lsum}\mathtt{:}\mathit{ZInc}},~\\
\qquad \phantom{ \;\|\; } \mathit{lsum} \leftarrow r_{k,1} * \lfloor b_{1,j} \rfloor + \ldots + r_{k,n} * \lfloor b_{n,j} \rfloor \\
\qquad  \;\|\; \mathtt{if}~r_{k,1} + \mathit{lsum} > c_k~\mathtt{then}~b_{1,j} \leftarrow (0,0) \\
\qquad  \;\|\; \ldots \\
\qquad  \;\|\; \mathtt{if}~r_{k,n} + \mathit{lsum} > c_k~\mathtt{then}~b_{n,j} \leftarrow (0,0)
\end{array}
\end{displaymath}
The entailment operations are given by:
\begin{displaymath}
\begin{array}{l}
\mathit{entailed}(x + y \leq c) \eqdef \lceil x \rceil + \lceil y \rceil \leq c \\[0.3cm]
\mathit{entailed}(\sum_{i \in [1..n]} r_{k,i} * b_{i,j} \leq c_k) \eqdef\\
\qquad r_{k,1} * \lceil b_{1,j} \rceil + \ldots + r_{k,n} * \lceil b_{n,j} \rceil \leq c_k
\end{array}
\end{displaymath}
We note that the partial order $\leq$ is viewed as a function returning $\mathit{true}$ whenever $a \leq b$ holds.
Logical connectors can also be defined in a similar manner, for instance we have:
\begin{displaymath}
\begin{array}{l}
\llbracket \varphi \land \psi \rrbracket \eqdef \llbracket \varphi \rrbracket \;\|\; \llbracket \psi \rrbracket \\[0.3cm]
\llbracket \varphi \Leftrightarrow \psi \rrbracket \eqdef \\
\phantom{ \;\|\; }\quad \mathtt{if}~\mathit{entailed}(\varphi)~\mathtt{then}~\llbracket \psi \rrbracket\\
 \;\|\; \quad \mathtt{if}~\mathit{entailed}(\psi)~\mathtt{then}~\llbracket \varphi \rrbracket\\
 \;\|\; \quad \mathtt{if}~\mathit{entailed}(\lnot\varphi)~\mathtt{then}~\llbracket \lnot\psi \rrbracket\\
 \;\|\; \quad \mathtt{if}~\mathit{entailed}(\lnot\psi)~\mathtt{then}~\llbracket \lnot\varphi \rrbracket
\end{array}
\end{displaymath}
We have defined various PCCP processes without proving the functions used were indeed monotone, as required by the PCCP model.
Automatically proving the monotonicity of arbitrary functions is a challenging task.
Rules are given by~\citet{entail-fd} for the restricted language of indexicals, and \citet{milano-towards-2012} discuss the topic on an extended indexicals language.
In PCCP, the next lemma provides a useful lattice-theoretic view of partial orders as monotone functions:
\begin{lemma}
Let $L$ be a lattice and $L^\partial$ its dual lattice, and $a \in L$ and $b \in L^\partial$.
The function $f: L \times L^\partial \to \mathit{BInc}$ defined as $f(a,b) = \mathit{true}$ iff $a \geq_L b$ and $\mathit{false}$ otherwise, is monotone.
\end{lemma}
\begin{proof}
$a' \geq_L a$ implies $f(a, b) \geq f(a', a)$ since $a' \geq_L a \geq_L b$.
Also, $b' \geq_{L^\partial} b$ (equivalent to $b' \leq_L b$) implies $f(a,b) \geq f(a,b')$ since $b' \leq_L b \leq_L a$.
\end{proof}
A particular case of this lemma is very useful to prove monotonicity of the functions in the previous examples.
When $c \in L$ is a constant and $a$ is a variable taking a value in $L$, a corollary is that $c \geq a$ (equivalently $a \leq c$) is monotone.
For instance, $\lceil x \rceil + \lceil y \rceil \leq c$ is monotone because the projection functions and the integer addition are monotone (and the functional composition of monotone functions is monotone as well).

Further discussion on monotone functions is out of scope of this paper, and we will suppose that all functions in PCCP processes are monotone.
Our focus is now on proving that PCCP programs can be correctly executed on a parallel hardware, and provide the same result than if executed on a sequential machine.

\section{Semantics and Correctness}

We introduce the denotational semantics, useful to prove properties on PCCP programs, the operational semantics, useful for implementation purposes, and the formal equivalence between the two.
A PCCP process is an extensive and monotone function over a Cartesian product $\mathit{Store} = L_1 \times \ldots \times L_n$ storing the values of all local variables.
Because PCCP does not have recursive definitions, the number of local variables is finite and known at compile-time.
Therefore, all local statements $\exists{x\mathtt{:}L},~P$ can be erased at compile-time by replacing each variable's name with an index in $\mathit{Store}$ representing the lattice element of this variable.
For instance, given a store $s \in \mathit{Store}$, we write $s(x) \eqdef \pi_x(s)$ the value of the variable $x$.
The three remaining statements are the ask, tell and parallel operations.

Our semantic treatment of PCCP is based on the one of CCP~\cite{saraswat-semantic-1991}.
However, instead of a constraint system, we manipulate a Cartesian product of lattices, allowing for a closer mapping between the theory and a practical implementation.
We first give the denotational semantics of PCCP by defining a function $\mathcal{D}: \mathit{Proc} \to (\mathit{Store} \to \mathit{Store})$:
\begin{displaymath}
\begin{array}{l}
\mathcal{D}(x \leftarrow f(y_1,..,y_n)) \triangleq \lambda s.\mathit{embed}_x(s, f(s(y_1),..,s(y_n))) \\
\mathcal{D}(\mathtt{if}~b~\mathtt{then}~P) \eqdef \lambda s.\mathit{if} s(b)~\mathit{then}~\mathcal{D}(P)(s)~\mathit{else}~s \\
\mathcal{D}(P\;\|\;Q) \eqdef \mathcal{D}(P) \sqcup \mathcal{D}(Q)
\end{array}
\end{displaymath}
The tell operation joins the result of the function with the component indexed $x$ of the store.
The ask operation verifies if the variable $b$ is $\mathit{true}$ in the store, in which case it maps to the denotation of $P$, and otherwise behaves like the identity function.
Finally the parallel composition is the join of the denotations of the two processes.
This mathematical definition seems to forbid $P$ and $Q$ to share the store $s$ as it is ``copied''.
When we introduce the operational semantics, this will not be the case anymore, but still, we will show that both semantics are equivalent.

Applying the function $\mathcal{D}(P)$ to the store might not yield a fixed point, and therefore it must be applied several times.
We characterize this function in the next theorem, which is largely due to~\citet{saraswat-semantic-1991}.
\begin{theorem}
Let $L_1,\ldots,L_n$ be complete lattices, and $\mathit{Store} = L_1 \times \ldots \times L_n$.
Then, for any process $P$, the least fixed point of $\mathcal{D}(P)$ exists and is unique.
Moreover, $\mathbf{fix}~\mathcal{D}(P)$ is a closure operator.
\label{closure-op}
\end{theorem}
\begin{proof}
$\mathcal{D}(P)$ is extensive and monotone as $\mathit{embed}$ and $\sqcup$ are extensive and monotone operations.
Moreover, $\mathit{Store}$ is a complete lattice as the Cartesian product of complete lattices yields a complete lattice.
By Tarski fixed point theorem, any monotone function on a complete lattice has a least fixed point.
This shows $\mathcal{D}(P)$ is completely described by its fixed points, thus $\mathbf{fix}~\mathcal{D}(P)$ is idempotent and is a closure operator since $\mathcal{D}(P)$ is already extensive and monotone.
\end{proof}
The claim made earlier that we can write PCCP programs with the same hypothesis than in sequential programs can be made more precise now.
Given a process $P$, if we replace all parallel compositions by sequential compositions (a transformation we write $\mathtt{seq}~P$), then the denotation of both programs remains the same.
We define sequential composition as:
\begin{displaymath}
\mathcal{D}(P\;\mathtt{;}\;Q) \eqdef \mathcal{D}(Q) \circ \mathcal{D}(P)
\end{displaymath}
\begin{proposition}
$\mathbf{fix}~\mathcal{D}(\mathtt{seq}~P) = \mathbf{fix}~\mathcal{D}(P)$
\label{denotational-seq}
\end{proposition}
\begin{proof}
Firstly, we notice that Theorem~\ref{closure-op} still holds for sequential programs as the functional composition preserves extensiveness and monotonicity.
Hence both $\mathcal{D}(\mathtt{seq}~P)$ and $\mathcal{D}(P)$ have a least fixed point, and we must show they coincide.
It boils down to proving the fixed points of $f \circ g$ equals the fixed points of $f \sqcup g$.
Because $f \leq f \circ g$ and $g \leq f \circ g$, necessarily we have $f \sqcup g \leq f \circ g$, and by monotonicity of $f$ and $g$, $\mathbf{fix}~f \sqcup g \leq \mathbf{fix}~f \circ g$.
The other direction is shown by $(f \sqcup g) \circ (f \sqcup g) \geq f \circ g$ which holds because $f \sqcup g \geq g$ and $f \circ (f \sqcup g) \geq f \circ g$.
Moreover, by monotonicity of $f$ and $g$, we have $(f \sqcup g) \circ (f \sqcup g) \geq f \circ g \Rightarrow \mathbf{fix}~(f \sqcup g) \circ (f \sqcup g) \geq \mathbf{fix}~f \circ g \Leftrightarrow \mathbf{fix}~f \sqcup g \geq \mathbf{fix}~f \circ g$.
Therefore, $\mathbf{fix}~f \sqcup g = \mathbf{fix}~f \circ g$.
\end{proof}
A practical implication of this proposition is that we are not forced to fully parallelize a PCCP program, and can choose, when appropriate, to use a sequential composition.
This choice depends on the number of threads available on a particular hardware, and how costly it is to start a thread.
Moreover, the proof highlights that executing once all processes sequentially always leads to a stronger store than executing once all processes in parallel.
In practice, however, a sequential iteration may take longer to complete than a parallel iteration since it can only be executed on a single core.

The operational semantics of PCCP is more compactly described on a \textit{guarded normal form} (GNF), which is obtained by lifting all parallel compositions at top-level.
We obtain a set of guarded commands of the form $\{b_1,\ldots,b_n\} \Rightarrow x \leftarrow f(y_1,...,y_m)$ that can be executed in parallel.
This program transformation is obtained by the following function:
\begin{displaymath}
\begin{array}{l}
\mathit{gnf}(A, x \leftarrow f(y_1,...,y_n)) \eqdef \{A \Rightarrow x \leftarrow f(y_1,...,y_n)\} \\
\mathit{gnf}(A, \mathtt{if}~b~\mathtt{then}~P) \eqdef \mathit{gnf}(A \cup \{b\}, P) \\
\mathit{gnf}(A, P\;\|\;Q) \eqdef \mathit{gnf}(A, P) \cup \mathit{gnf}(A, Q)
\end{array}
\end{displaymath}
This function essentially duplicates the ask conditions to lift nested parallel processes.
We note that a guarded command $\{b_1,\ldots,b_n\} \Rightarrow x \leftarrow f(y_1,...,y_n)$ is just a compact syntax for the PCCP program $\mathtt{if}~b_1~\mathtt{then}~\ldots~\mathtt{if}~b_n~\mathtt{then}~x \leftarrow f(y_1,...,y_n)$. 
Next, we show $\mathit{gnf}$ preserves fixed points.
\begin{proposition}
$\mathbf{fix}~\mathcal{D}(\mathit{gnf}(\{\},P)) = \mathbf{fix}~\mathcal{D}(P)$
\label{denotation-gnf}
\end{proposition}
\begin{proof}
The ask operation distributes over parallel composition, \emph{i.e.}, $\mathbf{fix}~\mathcal{D}(\mathtt{if}~b~\mathtt{then}~(P \;\|\; Q))$ equals $\mathbf{fix}~\mathcal{D}(\mathtt{if}~b~\mathtt{then}~P \;\|\; \mathtt{if}~b~\mathtt{then}~Q)$.
The full proof is obtained by structural induction over the processes in $\mathit{gnf}$.
\end{proof}
We describe the operational semantics of PCCP using a transition function $\hookrightarrow$ between states.
A state is a couple $\langle s, G \rangle$ where $s$ is the store and $G$ a set of guarded commands.
The operational semantics can be specified using a single rule on PCCP programs in GNF:
\begin{mathpar}
\inferrule[select]
{(\{b_1,\ldots, b_n\} \Rightarrow x \leftarrow f(y_1,...,y_m)) \in G \\ \bigwedge_{i \leq n} s(b_i)}
{\langle s, G \rangle \hookrightarrow \langle \mathit{embed}_x(s, f(s(y_1), \ldots, s(y_m))), G \rangle} \\
\end{mathpar}
The execution of a PCCP program is a possibly infinite sequence of states $\langle s_1, G \rangle \hookrightarrow \ldots \hookrightarrow \langle s_n, G \rangle \hookrightarrow \ldots$.
In each transition, we select one process to be executed.
The transition function is monotone and extensive because $\mathit{embed}_x$ is monotone and extensive, and $f$ and the condition $\bigwedge_{i \leq n} s(b_i)$ are monotone.

In the next section, we motivate why this seemingly sequential operational semantics is actually a correct representation of a low-level parallel execution.
Before that, we must discuss the nondeterminism of the rule \textsc{select} which allows us, at each step, to execute any process in $G$.
Nondeterminism is important to avoid forcing a particular scheduling strategy of the processes in the semantics.
The only requirement is fairness, otherwise the same process could be selected over and over again.
\begin{definition}[Fairness]
A scheduling strategy is fair if, for each process $P \in G$, it generates transitions such that:
\begin{displaymath}
\begin{array}{l}
\forall{i \in \mathbb{Z}},~\exists{j \in \mathbb{Z}},~j \geq i \;\land \\
\qquad \langle s_j, G\rangle \hookrightarrow \langle s_{j+1}, G \rangle \text{ selects the process $P$}
\end{array}
\end{displaymath}
\label{fairness-threads}
\end{definition}
A result of~\citet{cousot-chaotic-iterations-1977} on chaotic iterations, applied to constraint programming by~\citet{apt-essence-1999}, guarantees that the limit of all fair scheduling strategies coincide.
When the fixed point of $\hookrightarrow$ is reachable in a finite number of steps, which is the case if the lattices underlying the store are finite\footnote{Or satisfy the ascending chain condition, \emph{i.e.}, for all chains $x_1 \leq \ldots \leq x_n \leq \ldots$ there exists $i \in \mathbb{Z}$ such that $x_i = x_{i+1}$.}, the choice of a fair scheduling strategy has no impact on the result computed.
This is particularly important in practice as from one parallel hardware to another, or even from one execution to another, the scheduling strategy might change.

We denote by $\mathcal{O}(P)$ the function mapping a store $s$ to a store $s'$, such that $\langle s, \mathit{gnf}(\{\},P)\rangle \hookrightarrow \langle s', \mathit{gnf}(\{\},P)\rangle$ where $\hookrightarrow$ follows a fair scheduling strategy.
\begin{theorem}
$\mathbf{fix}~\mathcal{D}(P) = \mathbf{fix}~\mathcal{O}(P)$
\end{theorem}
\begin{proof}
By Proposition~\ref{denotation-gnf}, we consider the equivalent process $Q = \mathit{gnf}(\{\},P)$, and by Proposition~\ref{denotational-seq}, the denotation of the equivalent process $\mathtt{seq}~Q$.
The function $\mathcal{D}(\mathtt{seq}~Q)$ computes the functional composition of each guarded command in sequential order.
This corresponds to a particular fair scheduling strategy.
By~\citet{apt-essence-1999}, the fixed points of any fair execution of monotone functions coincide, and therefore it coincides with the operational execution $\mathbf{fix}~\mathcal{O}(P)$.
Therefore, $\mathbf{fix}~\mathcal{D}(P) = \mathbf{fix}~\mathcal{O}(P)$.
\end{proof}

\section{Load/Store Semantics}

We show our semantic correct w.r.t. a \textit{weak memory consistency model} (see, \emph{e.g.}~\cite{nagarajan-primer-consistency-2020}).
Consistency defines a correct behavior in terms of loads and stores in shared memory.
Load and store instructions are written $\mathbf{L}~a~r$ and $\mathbf{S}~r~a$ where $a$ is the location of a variable in shared memory and $r$ refers to a thread-local memory location such as a register.
For both instructions, the direction of the memory movement is from the first to the second argument.
To investigate the low-level details of execution, we compile a guarded command $\{b_1,\ldots, b_n\} \Rightarrow x \leftarrow f(y_1,...,y_m)$ to a sequence of loads and stores:
\begin{lstlisting}[language=caml,mathescape,morekeywords={L,S}]
while true then
    L $b_1~rb_1$; $\ldots$; L $b_n~rb_n$;
    if $rb_1 \land \ldots \land rb_n$ then
        L $y_1~ry_1$; $\ldots$; L $y_m~ry_m$;
        $\mathcal{C}(f(ry_1, \ldots, ry_m), \mathit{rf})$;
        L $x~rx$;
        $\mathcal{C}(rx \sqcup \mathit{rf}, ox)$;
        $\mathcal{C}(ox > rx, bx)$;
        if $bx$ then
            S $ox~x$;
\end{lstlisting}
The compilation function $\mathcal{C}(E,r)$ compiles the expression $E$, with its result stored into $r$.
As the corresponding expressions are only defined over thread-local variables, they do not pose any concurrent threats.
We note that $\sqcup$ and $>$ are compiled according to the lattice-type of the variable $x$.
Although the compilation of a guarded command is a sequential program, only a few of the program order dependencies are important.
\begin{lemma}
The only important program order ($\mathit{po}$) dependencies are $1 \xrightarrow{\text{po}} 2  \xrightarrow{\text{po}} 3 \xrightarrow{\text{po}} 9 \xrightarrow{\text{po}} 10$, $1 \xrightarrow{\text{po}} 6 \xrightarrow{\text{po}} 7$, and $1 \xrightarrow{\text{po}} 4 \xrightarrow{\text{po}} 5 \xrightarrow{\text{po}} 7 \xrightarrow{\text{po}} 8 \xrightarrow{\text{po}} 9$.
\label{po-guarded-command}
\end{lemma}
\noindent
This lemma can be used by compilers to improve efficiency by reordering some of the instructions.

We now define the load/store operational semantics.
The difference with the operational semantics of the previous section is that we have an assignment operator which does not perform a join in the store.
The assignment is written $s[a \mapsto b]$, meaning the value of $a$ is replaced by $s(b)$ in the store $s$.
We suppose the variable $r$ is local to a thread.
\begin{mathpar}
\inferrule[load]
{}
{\langle s, \mathbf{L}~a~r; P \rangle \twoheadrightarrow \langle s[r \mapsto a], P \rangle}

\inferrule[store]
{}
{\langle s, \mathbf{S}~r~a; P \rangle \twoheadrightarrow \langle s[a \mapsto r], P \rangle} \\

\inferrule[select2]
{I \in \mathit{GC} \\ \langle s, I \rangle \twoheadrightarrow \langle s', I' \rangle}
{\langle s, \mathit{GI} \rangle \Rightarrow \langle s', (\mathit{GI} \setminus \{I\}) \cup \{I'\} \rangle}
\end{mathpar}
We skip the operational semantics of \textbf{if} and \textbf{while} which is standard, and focus on the execution of a sequence of load and store instructions.
The rule \textsc{select2} is similar to \textsc{select} but executes instructions at a finer grain.
Each step of the transition $\Rightarrow$ models an atomic action on the shared memory.
The load/store operational semantics, and the proofs of correctness below, make the following assumptions on the memory consistency model and cache coherency protocol:
\begin{description}
\item[(\textbf{PO})] The program order of Lemma~\ref{po-guarded-command} must be enforced.
\item[(\textbf{ATOM})] Load and store instructions must be atomic.
\item[(\textbf{EC})] The caches must eventually become coherent.
\item[(\textbf{OTA})] Values cannot appear \textit{out-of-thin-air}.
\end{description}
\noindent
(\textbf{PO}) is merely a consequence of the Von Neumann model of computation.
(\textbf{ATOM}) is an important assumption to avoid torn reads and writes, \emph{e.g.}, a thread writes the two first bytes of a 32 bits integer and another thread writes the two last bytes.
Such a data race is undefined behavior in many languages, \emph{e.g.}, C++~\cite{ISO-IEC-14882:2020} (\S6.9.2.1) and CUDA~\cite{lustig-ptx-cuda-formal-2019}, which is why we must prevent it.
(\textbf{EC}) is relevant to cache coherency and supposes the threads eventually view the same value of any variable after a finite number of loads.
In other terms, a thread cannot work on its own view of the world forever, and communication must happen at one point in time.
(\textbf{OTA}) is a common assumption to avoid a thread to speculatively write a value which might be wrongly read by another thread~\cite{boehm-out-of-thin-air-2014}.
In addition, and although hardware vendors seldom specify the scheduling strategy employed, we make the additional fairness assumption (defined similarly to Def.~\ref{fairness-threads}):
\begin{description}
\item[(\textbf{FAIR})] The thread scheduling strategy must be fair.
\end{description}
We now prove soundness and completeness results.
Soundness implies the load/store semantics of a PCCP program does not generate more information than the operational semantics, and completeness that it does not generate less information.

\begin{theorem}[Soundness]
Let $\mathit{GI}$ be the load/store compilation of a PCCP program $P$.
Then for any sequence $\langle s_1, \mathit{GI} \rangle \Rightarrow \ldots \Rightarrow \langle s_i, \mathit{GI}' \rangle$, we have $s_i \leq (\mathbf{fix}~\mathcal{O}(P))(s_1)$.
\label{soundness}
\end{theorem}
\begin{proof}
By (\textbf{ATOM}), we must have $s_k(a) = s_k(r)$ after the execution of any instruction $\mathbf{S}~a~r$ or $\mathbf{L}~r~a$.
Moreover, the register $r$ can only evolve monotonically as it is local to a thread and is only modified by monotone operations.
Therefore, we have at each step $k$, $s_k(r) = s_k(a) \leq (\mathbf{fix}~\mathcal{O}(P))(s_i)(a)$ for all variables $a$.
By induction on $\Rightarrow$, it is necessary that $s_i \leq (\mathbf{fix}~\mathcal{O}(P))(s_1)$.
\end{proof}
Because of the assignment operator, the transition $\Rightarrow$ is neither extensive nor monotone.
However, we can still show that $\Rightarrow$ converges to the fixed point of $\mathcal{O}(P)$.
\begin{theorem}[Completeness]
Suppose the fixed point $\mathit{os} = (\mathbf{fix}~\mathcal{O}(P))(s_1)$ is reachable in a finite number of steps.
Then, $\exists{i \in \mathbb{N}},~\langle s_1, \mathit{GI} \rangle \Rightarrow \ldots \Rightarrow \langle s_i, \mathit{GI}' \rangle \Rightarrow \ldots$ such that $s_i \geq \mathit{os}$.
Further, for all $j \in \mathbb{N}$, $j > i$, $s_i = s_j$.
\label{completeness}
\end{theorem}
\begin{proof}
We first show that $\Rightarrow$ eventually makes \textit{progress}.
Because of the fairness condition (\textbf{FAIR}), we must reach a state $\langle s_k, \mathit{GI} \rangle$ such that all instructions have been executed at least once.
If no store operation was executed, then $s_k = \mathit{os}$ because all tell operations are at a fixed point.
Otherwise, at least one store operation on a variable $x$ has been executed, and due to the condition on lines 8--9 ($ox > rx$) and (\textbf{ATOM}), we necessarily have $s_k(x) > s_1(x)$.
This guarantees the progress of the transition $\Rightarrow$ after $k$ steps.

Now, if $s_k < \mathit{os}$, it means that $\mathcal{O}(P)$ is not at a fixed point on $s_k$, and therefore there is a transition $\langle s_k, G \rangle \hookrightarrow \langle s_{k+1}, G \rangle$ such that $s_k < s_{k+1}$.
Suppose the transition executes the process $\{b_1,\ldots,b_n\} \Rightarrow x \leftarrow f(y_1,\ldots,y_m)$.
By (\textbf{EC}), we must reach a state in which the latest values of $b_1,\ldots,b_n$ are visible to the thread, and the condition on line 3 is true.
Similarly for the values $y_1,\ldots,y_m$ in which case, by (\textbf{PO}), $f(ry_1,\ldots,ry_m)$ (line 5) and $rx \sqcup rf$ (line 7) must lead to the same result than the one given in the operational semantics.
At that point, $ox > rx$ must hold by hypothesis and a store operation must be issued.
Therefore, by induction on the finite sequence of operational transitions, we must have $s_i \geq \mathit{os}$.

The fact that $s_i = s_j$ for all $j > i$ is a consequence of Theorem~\ref{soundness}.
\end{proof}
Theorems~\ref{soundness} and~\ref{completeness} show that the result of the low-level parallel execution in terms of loads and stores coincides with the fixed point of the operational semantics, when it is reachable in a finite number of steps (which is the only behavior we are interested by).
The \textit{asynchronous iterations} of \citet{cousot-asynchronous-1977} generalizes this result, although in a more abstract mathematical setting, to the case of infinite sequences.

We address more practical aspects of the computation such as fixed point and failure detections in the next section.

\section{Turbo: GPU Constraint Solver}
\label{turbo}

\textsc{Turbo} is a constraint solver implementing \textit{entirely on GPU} a standard propagate and search algorithm~\cite{propagation-guido-tack}.
Some of the design choices of \textsc{Turbo} are direct consequences of the \textsc{Turing} and newer architectures of \textsc{nvidia} GPUs which we review now.
To illustrate our explanations, we take the example of the laptop GPU \textsc{nvidia Quadro RTX 5000 Max-Q}, which we later use to perform our experiments.
It has $48$ \textit{streaming multiprocessors} (SMs), each with $64$KB of L1 cache and $64$ cores.
There is a total of $48 * 64 = 3072$ cores.
When programming in CUDA, we group parallel processes in \textit{blocks} that are executed on a single SM, and have their own intra-block synchronization primitives.
In particular, we will use the barrier instruction \texttt{__syncthread()} which forces all threads of a block to wait for each other before continuing.
Threads in a block can communicate through up to $48$KB of \textit{shared memory}, a very fast on-chip memory local to a block.
Finally, we have $16$GB of \textit{global memory} shared among all SMs, which can be used for inter-blocks communication.

In order to increase the solving efficiency, the design of \textsc{Turbo} attempts to match the architecture of GPUs.
Firstly, we completely avoid communication latency with the CPU as the full solving algorithm is on the GPU.
Secondly, we dynamically generate subproblems following a variant of embarrassingly parallel search (EPS)~\cite{malapert-embarrassingly-2016}.
Each subproblem is then solved by a block on a single SM, which limits data movements between caches and global memory.
To maximize the utilization of the cache, avoid costly memory allocation and synchronization among threads, we have minimal data structures.
Propagation is \textit{eventless}~\cite{mackworth-consistency-1977}, hence we do not need to maintain a queue of propagators to be executed~\cite{schulte-efficient-2008} or take care of synchronization and workload balancing issues~\cite{parallel-consistency-cp}.
State restoration during backtracking is done by full recomputation~\cite{trail-vs-copy}.
Each block maintains two stores of variables: one for the root node of the subproblem and one for the computation.
On backtracking, the root node of the subproblem is copied in the current store, and the branching decisions of the new path to explore are committed into the store.
Therefore, we avoid maintaining a trailing data structure which would induce complicated thread synchronization.
Adaptative recomputation~\cite{trail-vs-copy} incurs copies of the stores which are not cache friendly, hence we chose to avoid extra copies completely.
This design is simple, yet efficient as shown below.

\citet{lustig-ptx-cuda-formal-2019} formalizes the memory consistency model of PTX, the instruction set architecture of \textsc{nvidia} GPUs.
The assumptions made by the load/store semantics are correct w.r.t. the PTX memory consistency model.
Indeed, \citet{lustig-ptx-cuda-formal-2019} introduce 6 axioms which match or strengthen our assumptions.
In particular, the \textit{SC-per-location axiom} implies (\textbf{PO}) as it states ``morally strong [...] communication order cannot contradict program order''.
The \textit{no-thin-air axiom} matches the (\textbf{OTA}) assumption, as both target the same issue.
Eventual cache coherence (\textbf{EC}) is enforced at the level of blocks because threads in a block share the same cache, therefore observe any store operation directly.
The remaining assumption (\textbf{ATOM}) is slightly more tricky.
When considering the CUDA programming model, store and load are not necessarily atomic.
Fortunately, \textsc{nvidia} recently implemented the C++11 atomic library for CUDA\footnote{\url{https://nvidia.github.io/libcudacxx/}}, and therefore we simply need to wrap shared integer and Boolean variables in a templated class \texttt{atomic<T>}.
\st{Another way to fulfil the (\textbf{ATOM}) assumption is to store the shared variables in the \textit{shared memory}.
Load and store accesses to shared memory are sequentialized\footnote{\url{https://docs.nvidia.com/cuda/cuda-c-best-practices-guide/index.html#shared-memory}}, and therefore satisfy the \textbf{(ATOM)} assumption.
However, this method is less generic than atomics and is proper to the hardware architecture of recent GPUs.
Nevertheless, in our experiments, we rely on the shared memory as we measured it to be the fastest solution.}
{\color{red}Erratum:} sequentialization does not imply atomicity, so one should solely relies on atomics, using relaxed load and store operations.
For the experiments, we measured a small 5\% decrease in efficiency when using atomics, so it does not change the conclusion of this paper.

\paragraph{Fixed point loop}

We now present the propagation loop with fixed point and failure detections, improving on the \texttt{while true} loop of the guarded commands compilation.
We define the propagation loop as a function \texttt{void propagation(int tid, int stride, VStore\& s, Vector<Propagator*>\& props)} where \texttt{tid} is a unique thread identifier, \texttt{stride} is the number of threads in the block, \texttt{s} is a store of interval variables and \texttt{props} is the array of propagators (which are PCCP processes):
\begin{lstlisting}[language=c++]
__shared__ bool has_changed[3] = {true};
for(int i=1; has_changed[(i-1)%3]; ++i){
  for (int t = tid; t < props.size(); t += stride){
    if(props[t]->propagate(s))
      has_changed[i%3] = true; }
  has_changed[(i+1)%3] = false;
  __syncthreads();
}
\end{lstlisting}
Fixed point detection is done using three Boolean variables \texttt{has_changed[3]}.
In each iteration, we manipulate the past, present and future values of \texttt{has_changed}.
First, we continue to iterate if a variable changed in the previous iteration (\texttt{has_changed[(i-1)\%3]}).
Second, we set the current value to true if a propagator successfully changes a domain (\texttt{has_changed[i\%3]}).
Note that for each guarded command, the Boolean value $bx$ indicates if a change occurred.
Thirdly, we must initialize to false the value of the next iteration.
At the end of each iteration, all threads reach a barrier, thus the current and next value of \texttt{has_changed} is correctly set.
As for failure detection, a failure is a fixpoint on $\top$, \emph{i.e.}, an empty interval in the store.
This case can be treated after the loop by examining \texttt{s}.

This propagation loop algorithm is reminiscent of the AC-1 algorithm~\cite{mackworth-consistency-1977} which iterates until arc consistency is achieved on all constraints.
The similarity only concerns the propagation loop as we do not control the level of consistency achieved by the propagators here.

\paragraph{Evaluation}

\begin{table}
\begin{tabular}{l c c c c}
solver & feas. & opt. & nodes-per-sec. & time (sec.) \\
\hline
\textbf{Patterson} & 110 & 110 & - & - \\
Turbo & 110 & 108 & 473k & 1362s \\
GeCode & 109 & 103 & 49k & 3812s \\
\hline
\textbf{j30} & 480 & 480 & - & - \\
Turbo & 379 & 309 & 423k & 5349s \\
GeCode & 376 & 307 & 63k & 5361s
\end{tabular}
\caption{Comparing Turbo to GeCode on the RCPSP scheduling problem.}
\label{result-table}
\end{table}

We evaluate \textsc{Turbo}\footnote{\textbf{Replicate}: the version of \textsc{Turbo} used in this paper is available at \url{https://github.com/ptal/turbo/tree/aaai2022}.} on the RCPSP problem introduced above.
As it is common knowledge in CUDA, it is best to over-saturate the SMs with blocks, and over-saturate the blocks with threads.
We have identified that multiplying by $4$ the SMs---to obtain $192$ blocks---and cores on SMs---to obtain $256$ threads---yield the best efficiency on this particular GPU model.
We compare \textsc{Turbo} with the well-known constraint solver \textsc{GeCode 6.2.0}~\cite{gecode} in parallel mode on a processor i7-10750@2.60GHz with 6 cores and 12 threads.
The constraint model and search strategy are the same for both solvers.
We experiment on 2 data sets: Patterson~\cite{patterson-comparison-1984} (timeout of 5 minutes) which has instances with various numbers of tasks and resources, and j30 from PSPSLIB~\cite{kolisch-psplib-project-1997} (timeout of 30 seconds) with 30 tasks and 4 resources.
The results are presented in Table~\ref{result-table}.
\textsc{Turbo} positively compares to GeCode, being slightly better on both data sets.
An important aspect is that \textsc{Turbo} processes up to an order of magnitude more nodes per seconds than GeCode.
Interestingly, GeCode failed on three instances due to an assertion relevant to the parallel code.
This shows the difficulty to code correctly multithreaded program, hence one advantage of the formally correct and simpler design of \textsc{Turbo}.
Our goal with these experiments was only to show a GPU-based solver can be competitive, but there are methods and solvers more efficient on RCPSP as shown by~\citet{schutt-rcpsp-2013}.

\section{Related Work and Discussion}

We are not aware of similar work connecting concurrent constraint programming and the low-level aspects of parallel programming.
Therefore, we primarily compare parallel constraint solvers to \textsc{Turbo}.

The literature on constraint solving on GPUs is scarce.
Actually, as far as we know, \textsc{nvidioso}~\cite{campeotto-exploring-2014} is the only propagate and search constraint solver running on GPU using CUDA.
A difference with \textsc{Turbo} is that \textsc{nvidioso} relies on the CPU for the backtracking and propagation of some constraints, whereas one guiding principle behind \textsc{Turbo} is to only rely on the GPU.
Moreover, in contrast to our work, the search component is not parallelized.
\textsc{phact} is another parallel constraint solver running on heterogeneous architectures (CPU, GPU and others) using \textsc{OpenCL}~\cite{roque-parallel-cp-phact-2018}.
However, it is not clear from the paper if and how propagation is parallelized.

A larger number of works is available when considering other solving algorithms or CPU parallelism~\cite{hamadi-handbook-2018}; we only overview a few selected ones.
Local search algorithms have been successfully implemented on GPUs, showcasing speed-up of one order of magnitude in comparison to sequential versions~\cite{arbelaez-gpu-2014,campeotto-gpu-2014}.
However, local search is not an exact method, which means there is no guarantee to find the best solution.
\citet{fioretto-accelerating-2018} successfully apply a \textit{dynamic programming} solving algorithm to GPUs.
Their results are very encouraging and lead to speed-up of two orders of magnitude.
However, as the authors warn us, this dynamic programming approach can require exponential time and space, and is only suited for specific problems.
In the case of CPU parallelism, the traditional approach is to parallelize the search component, see for instance~\citet{perron-search-1999} and~\citet{schulte-simple-pcp}.
\textsc{Turbo} relies on EPS which is a recent and robust search parallelization method~\cite{malapert-embarrassingly-2016}.
The class of \textit{distributed constraint satisfaction problems} is mostly relevant in a distributed computing context where agents communicate by message-passing~\cite{yokoo-algorithms-2000}, whereas we situate ourselves in a shared-state memory model.
We refer to the survey of \citet{gent-review-2018} for additional references.

It is worth noting that parallelizing propagation and search dates back to the eighties, with parallel logic programming languages.
Conjunctive and disjunctive logic predicates can be executed in parallel, similarly to the propagate and search components.
\citet{gupta-parallel-2001} survey the parallelization of logic programming languages.

Abstract interpretation~\cite{cousot-abstract-1977}, and its recent application to constraint reasoning~\cite{dsilva-abstract-2014,cousot-smt-product,pelleau-constraint-2013}, is deeply connected to PCCP, as both are based on lattice theory.
We have purposely avoided to introduce PCCP in the framework of abstract interpretation to keep the paper brief.
This framework would allow us to prove more properties, in particular when considering lattices that cannot be exactly represented by a machine; for instance the lattice of intervals of real numbers is approximated by the lattice of intervals of floating-point numbers.
Moreover, in the context of abstract interpretation, \citet{kim-deterministic-2020} introduce a method to compute deterministically a fixed point in parallel.
In our terms, they build a dependency graph among propagators and statically generate a scheduling of the propagators.
It improves our naive propagation loop where propagators are all executed, even those that are already at fixed points.
An adaptation of this work to constraint reasoning and GPUs is an interesting lead to improve the efficiency of the propagation loop.

\section{Conclusion}

We have laid the mathematical and practical foundation of PCCP, a new parallel programming language, useful for correctly computing fixed points in parallel.
We applied this paradigm to constraint solving with \textsc{Turbo} a parallel GPU-based constraint solver.
We experimentally validated \textsc{Turbo} on a scheduling problem, where the propagation and search are parallelized, and showed it competes positively with a CPU-based constraint solver.
This work is only a first step towards a full-fledged GPU-based constraint solver.
In particular, constraint learning techniques such as lazy clause generation~\cite{Ohrimenko:2009:PVL:1553323.1553342} pose additional challenges for their executions on GPUs, as already investigated in SAT solvers~\cite{hamadi-handbook-2018}.
Moreover, while some global constraints such as \texttt{table} natively support parallelization~\cite{campeotto-exploring-2014}, the applicability of PCCP for programming other global constraints remains to be shown.
We intend to explore a parallelization of the \texttt{range} and \texttt{roots} constraints which can construct many other global constraints~\cite{bessiere-range-roots-2009}.

\section*{Acknowledgments}

We are grateful to the reviewers for the useful comments, and the reference to the AC-1 algorithm.
This work is supported by the Luxembourg National Research Fund (FNR)---COMOC Project, ref. C21/IS/16101289.


\end{document}